\newcommand{\ifConferenceVersion}{\iftrue}
\newcommand{\ifJournalVersion}{\iffalse}
\newcommand{\comment}[1]{}

\ifConferenceVersion
\documentclass{article}[11pt,letter]

	\newcommand{\InConference}[1]{#1}
	\newcommand{\InJournal}[1]{}

	\usepackage[cmex10]{amsmath}	
	\usepackage{amsthm}
	\usepackage{caption}
\usepackage{fullpage}
		
	\newtheorem{definition}{Definition}[section]
	\newtheorem{lemma}[definition]{Lemma}
	\newtheorem{theorem}[definition]{Theorem}
	\newtheorem{proposition}[definition]{Proposition}
	\newtheorem{example}[definition]{Example}
	
	\newtheorem{observation}[definition]{Observation}

\else
	\documentclass[11pt]{article}

	\newcommand{\InConference}[1]{}
	\newcommand{\InJournal}[1]{#1}
	\usepackage{amsthm}
	\usepackage{fullpage}
	\usepackage[cmex10]{amsmath}	
	\newtheorem{definition}{Definition}[section]
	\newtheorem{lemma}[definition]{Lemma}
	\newtheorem{theorem}[definition]{Theorem}

	\newtheorem{observation}[definition]{Observation}
\fi

\usepackage[OT4]{fontenc}
\usepackage[english]{babel}
\usepackage{amssymb}
\usepackage{authblk}
\usepackage{mathtools}
\usepackage[utf8]{inputenc}
\usepackage{xspace}
\usepackage[breaklinks,bookmarks=false,hyperfootnotes=false]{hyperref}
\usepackage{graphicx}
\usepackage[T1]{fontenc}
\usepackage{algpseudocode}
\usepackage{algorithm}
\usepackage{booktabs}

\newcommand{\bigo}{\mathcal{O}}
\newcommand{\agent}{\ensuremath{\mathcal{A}}\xspace}

\newcommand{\port}{\ensuremath{\mathsf{port}}}

\newcommand{\vdeg}{\ensuremath{\mathsf{deg}}}

\newcommand{\RR}{{\textsc{Rotor-Router}}\xspace}

\newcommand{\state}{\mathcal{S}}
\newcommand{\mem}{\mathcal{M}}

\usepackage{setspace}   
\usepackage{color}

\begin{document}

\InConference
{

	\author[1]{Artur Menc}
	\affil[1]{Faculty of Fundamental Problems of Technology, Wroclaw University of Technology,  Poland}

	\author[2]{Dominik Paj\k{a}k\footnote{Corresponding author. E-mail dsp39@cl.cam.ac.uk}}
	\affil[2]{Computer Laboratory, University of Cambridge, UK}

	\author[3]{Przemys\l{}aw Uznański\thanks{Part of this work was done while D. Paj\k{a}k was visiting P. Uznański at  Aix-Marseille Université. Partially supported by the Labex Archimède and by the ANR project MACARON (ANR-13-JS02-0002)}}\affil[3]{Helsinki Institute for Information Technology HIIT, Department of Computer Science, Aalto University, Finland}
}

\InJournal
{
	\author[1]{Artur Menc}
	\affil[1]{Faculty of Fundamental Problems of Technology, Wroclaw University of Technology,  Poland}

	\author[2]{Dominik Paj\k{a}k}
	\affil[2]{Computer Laboratory, University of Cambridge, UK}

	\author[3]{Przemys\l{}aw Uznański\thanks{Part of this work was done while D. Paj\k{a}k was visiting P. Uznański at  Aix-Marseille Université. Partially supported by the Labex Archimède and by the ANR project MACARON (ANR-13-JS02-0002)}}\affil[3]{Helsinki Institute for Information Technology HIIT, Department of Computer Science, Aalto University, Finland}
}

\title{Time and space optimality of rotor-router graph exploration}
\maketitle

\begin{abstract}
	We consider the problem of exploration of an anonymous, port-labeled, undirected graph with $n$ nodes and $m$ edges and diameter $D$, by a single mobile agent. Initially the agent does not know the graph topology nor any of the global parameters. Moreover, the agent does not know the incoming port when entering to a vertex. Each vertex is endowed with memory that can be read and modified by the agent upon its visit to that node. However the agent has no operational memory i.e., it cannot carry any state while traversing an edge. In such a model at least $\log_2 d$ bits are needed at each vertex of degree $d$ for the agent to be able to traverse each graph edge. This number of bits is always sufficient to explore any graph in time $O(mD)$ using algorithm \RR~\cite{Yanovski}. We show that even if the available node memory is unlimited then time $\Omega(n^3)$ is sometimes required for any algorithm. This shows that \RR is asymptotically optimal in the worst-case graphs. Secondly we show that for the case of the path the \RR attains exactly optimal time.


	\end{abstract}




\section{Introduction}
\label{sec:intro}
In this paper we consider the exploration problem of a port-labeled graph in the following setting. The exploration is performed by a single agent that has no memory (oblivious agent) and when it enters to a node it has no information about the port number through which it entered (no inport). Each node contains some number of bits of memory that can be read and modified by the agent upon its visit. Hence the whole navigation mechanism needs to be defined using only local information. 

\RR is an algorithm in which each node maintains a pointer to one of its neighbors and a cyclic sequence of its neighbors. Upon each visit of the token to a node, the token is propagated along the pointer of its current node and the pointer of the node is advanced to the next position in the cyclic sequence. Studies of \RR show that it can be used as a graph-exploration algorithm and there are guarantees on exploration time even is the initial state of the pointers and the sequence at each node can be set by an adversary. In such a case, \RR can be also seen as a space-efficient algorithm since the only information that needs to be stored at a node is the current position of the pointer. Hence $\lceil \log_2 d \rceil$ bits of memory needs to be stored at each node with degree $d$ to implement \RR. It is worth observing that with less memory the task becomes impossible as at each node with degree $d$, the algorithm needs at least $d$ different inputs to traverse every outgoing edge. This paper considers the following question: does there exist an algorithm for exploration for oblivious agent with no inport that is always faster than \RR if we allow more bits of memory at each node? We give a negative answer to this question by showing a graph for which any such algorithm cannot be asymptotically faster than \RR even if we allow unbounded memory at each node.

\subsection{Preliminaries}
 We assume anonymous graph $G = (V,E)$ with $n$ nodes, $m$ edges, diameter $D$ and with no global labeling of the nodes nor the edges. In order to navigate in the graph, the agent needs to locally distinguish between the edges outgoing from its position, so we assume that all edges outgoing from a fixed vertex with degree $d$ are distinctly labeled using numbers $\{1,2,\dots,d\}$. 

The agent is modeled as a memoryless token. Each node contains a label of $\mem_w$ bits which can be read and modified by the agent upon its visit to that node. Such a model will be referred to as \emph{an oblivious agent}. Memory size on each node can depend on the node degree.

Let us denote by $\state_w = \{0,1,\dots,2^{\mem_w}-1\}$ the set node states. An oblivious agent is then defined as a function $f:\state_w \times \mathbb{N} \rightarrow \state_w \times \mathbb{N}$ whose input is a tuple $(s_w,d) \in \state_w \times \mathbb{N}$, where $s_w$ is the state of the node currently occupied by the agent and $d$ is the node degree. The output of function $f$ is a tuple $(s'_w,p) \in \state_w \times \{1,2,\dots,d\}$, where $s'_w$ is the new state of the currently occupied node and $p$ is the port number through which the agent exits the current node in the current step. We say that the agent is located at a node $v$ \textit{at the beginning} of some step $t$, then traverses the chosen arc \textit{during} step $t$ and appears at the other end of the arc at the beginning of step $t+1$. Initially, each node is in a starting state $s^{0}_w$.

Observe that the port label through which the agent entered to the current node is not part of the input. Thus the agent cannot easily backtrack its moves. We call this model feature \textit{an unknown inport}. 

At any step, the configuration of \RR is a triple $((\pi_v)_{v\in V},(\rho_v)_{v\in V},r)$, where $r$ is the current position of the agent. In the current step, the agent traverses the edge indicated by pointer $\pi_r$ which is then advanced to the next edge in the port ordering $\rho_r$.

\paragraph{Graph exploration problem.}
The goal of the agent is to visit all vertices of graph $G$. We assume that the initial agent position in the graph as well as the port-labeling of the edges can be chosen by an adversary. Initially the agent has no knowledge about the topology of $G$ or even its size. 


An oblivious agent that explores all unknown graphs, needs to traverse all edges and thus needs at least $d$ different inputs at any vertex with degree $d$ in order to traverse all of its outgoing edges. This leads to the following lower bound on the total memory at a vertex and on the agent.

\begin{observation}
\label{trivial_lower_bound}
If an oblivious agent explores all graphs in the model with unknown inport, then $\mem_w \geq \log_2 d$ holds for any node with degree $d$.
\end{observation}



\subsection{Our results}
In this paper we show two lower bounds. In the first we show that for any oblivious algorithm with no inport there exists a labeling of a path for which this algorithm needs at least $(n-1)^2$ steps. This shows that unbounded memory at a node cannot decrease the exploration time of the path even by one step. On the other hand \RR needs only $1$ bit and always explores a path in time at most $(n-1)^2$.
For general graphs we show that any oblivious agent in the model with no inport requires time $\Omega(n^3)$ for some graphs regardless of the sizes of node memory. This shows that it is impossible to construct an algorithm for oblivious agent that would be asymptotically faster than the \RR in the worst-case even if unbounded memory at each node is available.

\subsection{Related work}
  When exploring a graph using a \RR mechanism with arbitrary initialization, time $\Theta(mD)$ is always sufficient and sometimes required for any graph~\cite{Yanovski}~\cite{LockIn}. Since the \RR requires no special initialization, it can be implemented in a graph with $\lceil\log_2 d\rceil$-bits of memory at each node with degree $d$. An oblivious agent can simply exit the node $v$ via port $w(v)+1$, where $w(v)$ is the value on the whiteboard, and increment the value $w(v)$ modulo $\vdeg(v)$. Thus exploration in time $\bigo(mD)$ is possible by oblivious agents with $\lceil\log_2 d\rceil$-bit node memory, which is the smallest possible by Observation~\ref{trivial_lower_bound}.

\section{Lower bounds}
\label{sec:lower}
In this section we prove a lower bound on the number of steps of graph exploration for oblivious agent. First we need the following observation, which helps to reason about behavior of oblivious agents in port-labeled graphs.

\begin{lemma}
\label{lem:behaviour}
Behavior of any oblivious agent $\agent$ in graph $G$ with arbitrary size of node memory is fully characterized by the collection of functions $\port_d(i)$ for $d=1,2,\dots$. For a fixed $d$, the function denotes the outport chosen by the agent upon its $i$-th visit to any node with degree $d$.
\end{lemma}
\begin{proof}
Since the agent has no internal memory and does not know the label of the port through which it enters to a node, the only information the agent has is the degree of the current node and the state of the node memory. Thus the label of the next outport taken from a node $v$ can only depend on the degree of $v$ and on the labels of previously chosen outports from node $v$. 
\end{proof}

The following lemma characterizes the worst-case exploration time for any oblivious agent on a path.

\begin{theorem}
\label{thm:path}
Let $P$ be a path on $n$ vertices $v_1,v_2,\ldots,v_n$. For any oblivious agent $\agent$, starting on vertex $v_n$, there exists port labeling of $P$ such that if $\agent$ visits $v_1$, then:
\begin{enumerate}
\item \label{path1} $\agent$ makes at least $(n-1)^2$ edge traversals before its first visit to $v_1$,
\item \label{path2} $\agent$ traverses arc $v_n \rightarrow v_{n-1}$ at least $n-1$ times before its first visit to $v_1$.
\end{enumerate}
\end{theorem}
\begin{proof}
Let us fix any oblivious agent $\agent$. By Lemma~\ref{lem:behaviour} we can also fix its sequence $\{a_i\}_{i=1}^{\infty}$ of exits from each node of degree $2$ (that is $a_i = \port_2(i) \in \{1,2\}$).

The proof will proceed by induction on the number of nodes of the path. For $i=2$, the agent makes one traversal of $v_2 \rightarrow v_1$ before its first visit to $v_1$. 

Let us assume, that the claim is true for some $i\geq 2$ and we will show it for $i+1$. The agent starts at node $v_{i+1}$. We will use the inductive assumption for subpath of nodes $v_{i},\dots,v_2,v_1$. The agent needs to traverse arc $v_{i} \rightarrow v_{i-1}$ at least $i-1$ times before its first visit in $v_1$. We want to choose port labeling of node $v_{i}$. Take first $2(i-1)-1$ elements of sequence $\{a_i\}$ and take the element $\alpha$ that appears in it at least $(i-1)$ times. 
Set $\alpha$ to be the label of the arc $v_{i} \rightarrow v_{i+1}$. The other element, $\{1,2\} \setminus \{\alpha \}$ is the label of the arc $v_{i} \rightarrow v_{i-1}$. Observe that under such port labeling, before $i-1$ traversals of $v_{i} \rightarrow v_{i-1}$, arc $v_{i} \rightarrow v_{i+1}$ will be traversed at least $i-1$ times. This means that before the first visit to $v_1$, the agent will enter node $v_{i+1}$ at least $i-1$ times. Since the agent started at $v_{i+1}$ then before the first visit to $v_1$ it will traverse edge $v_{i+1} \rightarrow v_i$ at least $i$ times. The total number of steps within subpath $v_i,\dots v_1$ is at least $(i-1)^2$ by the inductive assumption. Additionally the agent traverses $v_{i+1} \rightarrow v_i$ at least $i$ and $v_i \rightarrow v_{i+1}$ at least $i-1$ times thus the total number of traversals is at least $(i-1)^2+(i-1)+i = i^2$.
\end{proof}
The previous theorem showed that any strategy for oblivious agents requires at least $(n-1)^2$ steps to explore a path. Interestingly, $(n-1)^2$ is also the worst-case number of steps to explore a path for the \RR. It shows that the \RR is optimal on a path. It means that even adding arbitrarily large node memory cannot provide any speedup (not only asymptotic one) for path exploration, when compared to just one bit node memory sufficient to implement \RR on the path.

The next theorem shows a lower bound on exploration time for oblivious agents on arbitrary graphs.
\begin{theorem}
\label{thm:lower_general}
For any value of $n$ and any any oblivious agent $\agent$ there exists graph $G$ with $n$ vertices, such that $\agent$ needs at least $\Omega(n^3)$ time steps to visit all vertices of $G$.
\end{theorem}
\begin{proof}
Fix agent $\agent$ and size of the graph $n$. Assume that $n$ divisible by $3$. 
Let $d = n/3$. Consider the sequence $a_i = \port_{d}(i) \in \{1,2,\dots d\}$ for $i=1,2,\dots$, where $\port_{d}(i)$ is defined in Lemma~\ref{lem:behaviour}. In the prefix of length $d \cdot (d-1) $ of sequence $\{a_i\}$ there exists a value $p$ that appears at most $d-1$ times. 

Consider a graph $G_1$ constructed from the clique $K_{d}$ by attaching one node $v'$ to each node $v$ of the clique (see Figure~\ref{fig:G1} for an example). Observe that each node coming from clique $K_{d}$ has an additional neighbor thus its degree in $G_1$ is $d$. For each $v$ from the clique, the port leading to $v'$ is $p$. All other ports are set arbitrarily.

\begin{figure}[ht]
\centering
\begin{minipage}{.5\textwidth}
  \centering
  \includegraphics[width=0.65\linewidth]{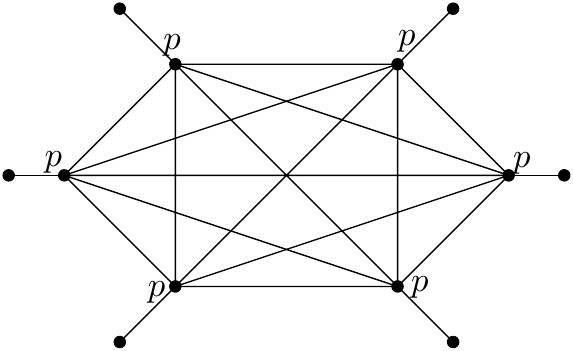}
  \captionof{figure}{Example of graph $G_1$ for $n = 18$.}
  \label{fig:G1}
\end{minipage}%
\begin{minipage}{.5\textwidth}
  \centering
  \includegraphics[width=\linewidth]{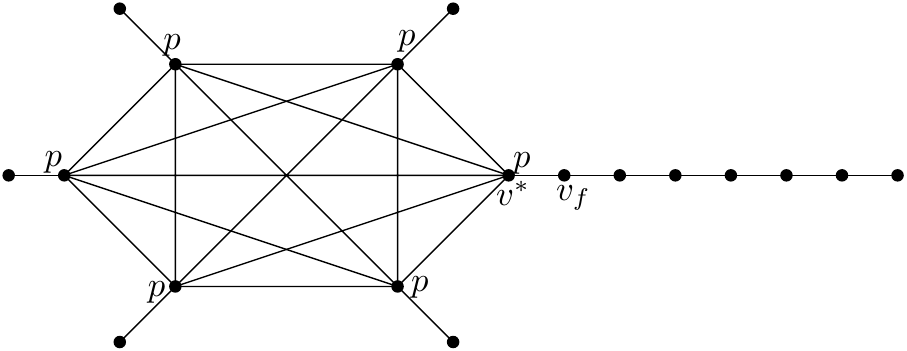}
  \captionof{figure}{Example of graph $G$ for $n = 18$.}
  \label{fig:G}
\end{minipage}
\end{figure}

Consider a walk of agent $\agent$ on graph $G_1$ starting from an arbitrary vertex $v_s$ for $d^2 \cdot (d-1)$ steps. There exists a vertex $v^*$ from the original clique $K_{d}$ that was visited at most $d \cdot (d-1)$ times. 

Construct a graph $G$ by modifying $G_1$. Replace the one additional node attached to $v^*$ with a path $P$ of $d+1$ nodes. Set the worst-case port-labeling of path $P$, as in the Theorem~\ref{thm:path}, depending on function $\port_2(i)$ of $\agent$. Denote by $v_f$, the first node of $P$ that is connected to $v^*$ (see Figure~\ref{fig:G} for an example).

Consider agent $\agent$ exploring the graph $G$ starting from vertex $v_s$. Since the agent is oblivious, its moves between vertices in $G$ that come from original graph $G_1$ are the same as in the graph $G_1$. Thus within $d^2 \cdot (d-1)$ steps in $G$, node $v^*$ is visited at most $d \cdot (d-1)$ times. Since $p$ is the port leading from $v^*$ to $v_f$ then after $d \cdot (d-1)$ visits to $v^*$, agent visited $v_f$ at most $d-1$.
But by Theorem~\ref{thm:path}, the agent needs to visit $v_f$ at least $d$ times to explore path $P$. Thus the agent needs time at least $d^2 \cdot (d-1) = \Omega(n^3)$ to explore graph $G$.

If $n$ is not divisible by $3$ we can add the remaining vertices to the path and the exploration time will be at least $\left\lfloor \frac{n}{3} \right\rfloor^2 \cdot \left(\left\lfloor \frac{n}{3} \right\rfloor-1\right) = \Omega(n^3)$
\end{proof}
The theorem shows that, even with unbounded node memory, the oblivious agents need $\Omega(n^3)$ steps to explore some graphs. Since the \RR explores any graph in time $\bigo(mD) = \bigo(n^3)$~\cite{Yanovski} there is no strategy for oblivious agents that would be faster in the worst-case. Observe also that the \RR can be implemented using node memory of minimum possible size $\lceil \log_2 d \rceil$ at nodes of degree $d$. By Observation \ref{trivial_lower_bound}, agent with less memory cannot traverse all outgoing edges. Thus the \RR is both time and space optimal strategy for oblivious agents.

\bibliographystyle{abbrv}
\bibliography{one-bit}



\end{document}